\documentclass[lettersize,journal]{IEEEtran}
\usepackage{cite}
\usepackage{tumcolor}
\usepackage{amsmath,amssymb,amsthm,fixmath}
\usepackage{mathtools}
\usepackage{algorithm}
\usepackage{algorithmic}
\usepackage{xcolor}
\usepackage{bm}
\usepackage{makecell}
\usepackage{pgfplots}
\usepackage{tikz}
\usepackage{graphicx}
\usepackage{amsmath}
\usetikzlibrary{intersections,pgfplots.fillbetween,patterns,spy,shapes.geometric,calc,positioning}
\usetikzlibrary{arrows,decorations.pathmorphing,backgrounds,fit,matrix,fillbetween}
\usetikzlibrary{shapes.geometric, arrows.meta}
\usetikzlibrary{decorations.pathreplacing,decorations.markings}
\usepgfplotslibrary{groupplots}
\usepackage{subcaption}
\usepackage{tabularx}
\usepackage{soul}
\usepackage[acronym,shortcuts]{glossaries}
\usepackage{multirow}
\usepackage{multicol}
\usepackage{makecell}
\usepackage{hhline}
 \usepackage{url}
 \usepackage{float}
 \usepackage{adjustbox}
 \usepackage{siunitx} 

\newtheorem{theorem}{Theorem}

\newtheorem{assumption}{Assumption}

\definecolor{Purple}{rgb}{0.6350 0.0780 0.1840}
\definecolor{DarkGreen}{rgb}{0 0.50 0}
\definecolor{BrightGreen}{rgb}{0.71, 0.89, 0.58}

\usepgfplotslibrary{external} 
\newacronym{MMSE}{MMSE}{minimum mean squared error}
\newacronym{MSE}{MSE}{mean squared error}
\newacronym{NMSE}{NMSE}{normalized MSE}
\newacronym{LMMSE}{LMMSE}{linear minimum mean squared error}
\newacronym{WSS}{WSS}{wide sense stationary}
\newacronym{DFT}{DFT}{discrete Fourier transform}
\newacronym{EM}{EM}{expectation-maximization}
\newacronym{BIC}{BIC}{Bayesian information criterion}
\newacronym{AIC}{AIC}{Akaike information criterion}
\newacronym{i.i.d.}{i.i.d.}{independent and identically distributed}
\newacronym{FLOP}{FLOP}{floating point operation}
\newacronym{AWGN}{AWGN}{additive white Gaussian noise}
\newacronym{LS}{LS}{least squares}
\newacronym{SCM}{SCM}{sample \ac{CM}}
\newacronym[firstplural = {covariance matrices (CMs)}]{CM}{CM}{covariance matrix}
\newacronym{ICM}{ICM}{inverse \ac{CM}}
\newacronym{AR}{AR}{autoregressive}
\newacronym{MA}{MA}{moving-average}
\newacronym{ARMA}{ARMA}{autoregressive moving-average}
\newacronym{FBM}{FBM}{fractional Brownian motion}
\newacronym{FFT}{FFT}{fast Fourier transform}
\newacronym{GMM}{GMM}{Gaussian mixture model}
 \newacronym{VAE}{VAE}{variational autoencoder}
 \newacronym{NN}{NN}{neural network}
 \newacronym{PD}{PD}{positive definite}
 \newacronym{OP}{OP}{optimization problem}
\newacronym{GS}{GS}{Gohberg-Semencul}
\newacronym{CSI}{CSI}{channel state information}
\newacronym{SAVG}{SAVG}{SCM averaged along its diagonals}
\newacronym{UE}{UE}{user equipment}
\newacronym{DoA}{DoA}{direction of arrival}
\newacronym{DoAs}{DoAs}{directions of arrival}
\newacronym{DoDs}{DoDs}{directions of depature}
\newacronym{ML}{ML}{machine learning}
\newacronym{BS}{BS}{base station}
\newacronym{WSSUS}{WSSUS}{wide-sense-stationary-uncorrelated-scattering}
\newacronym{ULA}{ULA}{uniform linear array}
\newacronym{URA}{URA}{uniform rectangular array}
\newacronym{MIMO}{MIMO}{multiple-input-multiple-output}
\newacronym{OFDM}{OFDM}{orthogonal-frequency-division-multiplexing}

\newacronym{Eig}{Eig}{eigenvalue}
\newacronym{Frob}{Frob}{Frobenius}
\newacronym{PGD}{PGD}{projected gradient descent}
\newacronym{PLS}{PLS}{projected LS}
\newacronym{LOS}{LOS}{line-of-sight}
\newacronym{GAN}{GAN}{generative adversarial network}
\newacronym{NLOS}{NLOS}{non-line-of-sight}
\newacronym{SISO}{SISO}{single-input-single-output}
\newacronym{SNR}{SNR}{signal-to-noise ratio}
\newacronym{PG}{PG}{probabilistic graph}
\newacronym{BN}{BN}{Bayesian network}

\DeclareMathOperator{\E}{\mathbb{E}}
\newcommand{\ldotsTwo}{%
  \mathinner{{\ldotp}{\ldotp}}%
}
\begin{document}

\title{A Statistical Characterization of Wireless Channels Conditioned on Side Information}

\author{Benedikt B\"ock, \IEEEmembership{Graduate Student Member, IEEE}, Michael Baur, \IEEEmembership{Graduate Student Member, IEEE},\\ Nurettin Turan, \IEEEmembership{Graduate Student Member, IEEE}, Dominik Semmler, \IEEEmembership{Graduate Student Member, IEEE},\\ and Wolfgang Utschick, \IEEEmembership{Fellow, IEEE}
        % <-this % stops a space
\thanks{\textcopyright This work has been submitted to the IEEE for possible publication.
Copyright may be transferred without notice, after which this version may no
longer be accessible.}% <-this % stops a space
%\thanks{Manuscript received April 19, 2021; revised August 16, 2021.}
}

\maketitle

\begin{abstract}
Statistical prior channel knowledge, such as the \ac{WSSUS} property, and additional side information both can be used to enhance physical layer applications in wireless communication. Generally, the wireless channel's strongly fluctuating path phases and \ac{WSSUS} property characterize the channel by a zero mean and Toeplitz-structured covariance matrices in different domains. In this work, we derive a framework to comprehensively categorize side information based on whether it preserves or abandons these statistical features conditioned on the given side information. To accomplish this, we combine insights from a generic channel model with the representation of wireless channels as probabilistic graphs. Additionally, we exemplify several applications, ranging from channel modeling to estimation and clustering, which demonstrate how the proposed framework can practically enhance physical layer methods utilizing \ac{ML}. 
\end{abstract}


\section{Introduction}

Statistical knowledge about the first and second moment of the wireless channel between a \ac{UE} and a \ac{BS} is of key importance to improve their communication link. Among other things, the first and second moment of the channel can be used to improve channel estimation \cite[Sec. 3.5.2]{Tse2005}, provide information about channel parameters \cite[Sec. 2.6]{bjoernson2017}, or can be employed in other applications, where instantaneous \ac{CSI} is not available or costly to acquire, with the advantage of reduced pilot and computational overhead \cite{Vagenas2011}. Hence, structural prior information about the channel's first and second moment can be beneficial, e.g., to reduce the number of required samples for estimating these moments \cite{boeck2023}. A prominent example of such prior information is the \ac{WSSUS} assumption, which characterizes the wireless channel as a \ac{WSS} process in both the temporal and frequency domains and, thus, constrains the channel \acp{CM} in these domains to be Toeplitz structured \cite{bello1963}. Equivalently, the \ac{WSSUS} assumption can be extended to the spatial domain, leading to a spatial \ac{CM} with Toeplitz structure \cite[Sec. 2.6]{yin2016}.

Recently, leveraging additional side information about the wireless channel between the \ac{UE} and the \ac{BS} to enhance their communication link has attracted a lot of attention in research. While this side information either can be interpretable in form of the \ac{UE}'s position \cite{Zhang2021}, it can also be given in an abstract representation by some \ac{ML}-based latent embedding~\cite{Studer2018}. Depending on the side information characteristics, the conditioning on this information can either preserve or abandon structural features of the channel's mean and \ac{CM}. In this work, we aim to establish a comprehensive framework for characterizing any side information based on its influence on the first and second channel moments. Our main contributions are the following:

\begin{itemize}
    \item We establish a theorem leveraging the statistical relation between arbitrary side information and the complex path loss phases of a channel to describe how the \ac{WSSUS} and zero-mean channel properties are either preserved or abandoned given this side information.
    \item By combining this theorem with a probabilistic graph representation of wireless channels, we introduce a framework, which allows to comprehensively categorize side information regarding its effect on the channel's \ac{WSSUS} and zero-mean properties.
    \item We present various exemplary applications of this framework. Specifically, we introduce a validation technique for the proper training of \ac{ML}-based channel models, regularize channel clustering and analyze the utility of side information for channel estimation by means of our proposed framework.
\end{itemize}

\section{Preliminaries}
\subsection{Channel Model}
\label{sec:channel_model}

Over the last decades, several channel models have been proposed following different paradigms and imposing slightly different assumptions \cite{almers2007}. In this work, we consider a generic wideband and time-varying \ac{MIMO} baseband channel, which is sampled equidistantly in the time, frequency as well as the spatial domains. This is the case in the typical \ac{OFDM} setup with constant subcarrier spacings and constant symbol durations, in which the transmitter and the receiver are both equipped with \acp{ULA}. The resulting channel (tensor) is given by
\begin{equation}
    \label{eq:channel_tensor}
    \bm{H} = \sum_{\ell=1}^L \sqrt{p_\ell}\mathrm{e}^{-\operatorname{j}\beta_\ell} \bm{a}_f(\tau_\ell) \otimes \bm{a}_t(\nu_\ell) \otimes \bm{a}_{\operatorname{R}}(\theta^{(\operatorname{R})}_\ell) \otimes \bm{a}_{\operatorname{T}}(\theta^{(\operatorname{T})}_\ell)
\end{equation}
characterized by $L$ paths and the channel parameters, i.e., its complex path losses $\sqrt{p_\ell}\mathrm{e}^{-\operatorname{j}\beta_\ell}$, delays $\tau_\ell$, Doppler-shifts $\nu_\ell$, \ac{DoAs} $\theta^{(\operatorname{R})}_\ell$, and \ac{DoDs} $\theta^{(\operatorname{T})}_\ell$. The vectors $\bm{a}_{(\cdot)}(\cdot)$ denote the steering vectors across the different domains, respectively. 
The path loss phase shifts $\{\beta_\ell\}_{\ell=1}^L$ contain polarization effects as well as the center frequency phase shift $2\pi d_\ell/\lambda_c$ with center wavelength $\lambda_c$ and path distance $d_\ell$. Our theoretical findings in Section \ref{sec:main_result} build on the following model assumptions.
\begin{assumption}
\label{as:dl}
The phases $\beta_\ell$ are uniformly distributed, i.e.,
\begin{equation}
    \label{eq:beta_uniform}
    \beta_\ell \sim \mathcal{U}(0,2 \pi)\ \text{for}\ \text{all}\ \ell.
\end{equation}
\end{assumption}
Assumption \ref{as:dl} holds when the path distances $d_\ell$ are not known on a scale of the center wavelength $\lambda_c$ and $d_\ell \gg \lambda_c$ \cite[Sec. 2.4.2]{Tse2005}. This is the case in typical wireless communication scenarios rendering Assumption \ref{as:dl} to be generally reasonable.
\begin{assumption}
\label{as:independence}
The phases $\beta_\ell$ are statistically independent of all channel parameters as well as across different paths, i.e., \begin{align}
    \label{eq:beta_assumption}
    & \beta_\ell \perp \kappa_{\ell'}\ \text{for}\ \text{all}\ \ell, \ell'\ \text{and}\ \kappa \in \{p,\tau,\nu,\theta^{(\operatorname{R})},\theta^{(\operatorname{T})}\}, \\
    \label{eq:para_assumption}
   & \beta_\ell \perp \beta_{\ell'}\ \text{for}\ \text{all}\ \ell \neq \ell'.
\end{align}
\end{assumption}
Although rigorously, the phases $\{\beta_\ell\}_{\ell=1}^L$ depend on the delays $\tau_\ell$, they are commonly modeled to satisfy Assumption \ref{as:independence} due to their strong fluctuations and the multitude of different influential effects contained \cite[Sec. 7.5]{3gpp}.
In addition, within small frequency, time and spatial ranges, the channel parameters and steering vectors are constant over the frequency $f$, the time $t$ and the positions of the transmitting and receiving antennas. As a result, the channel exhibits a zero mean and the \ac{WSSUS} property, i.e., it is a \ac{WSS} process across all domains~\cite{bello1963}. This implies $\bm{H}$ in \eqref{eq:channel_tensor} to possess a zero mean and a Toeplitz structured \ac{CM} in any domain. The steering vectors are given by 
\begin{align}
    \label{eq:a_tau}
    \bm{a}_f(\tau_\ell) = [1,\mathrm{e}^{-\operatorname{j}2\pi\Delta f\tau_\ell},\ldotsTwo,\mathrm{e}^{-\operatorname{j}2\pi\Delta f(M_{\mathrm{SC}}-1)\tau_\ell}]^{\operatorname{T}},\\
    \label{eq:a_nu}
    \bm{a}_t(\nu_\ell) = [1,\mathrm{e}^{-\operatorname{j}2\pi\Delta T\nu_\ell},\ldotsTwo,\mathrm{e}^{-\operatorname{j}2\pi\Delta T(M_{\mathrm{SN}}-1)\nu_\ell}]^{\operatorname{T}},\\
    \label{eq:a_R}
    \bm{a}_{\operatorname{R}}(\theta^{(\operatorname{R})}_{\ell}) = [1,\mathrm{e}^{-\operatorname{j}\pi\sin(\theta^{(\operatorname{R})}_{\ell})},\ldotsTwo,\mathrm{e}^{-\operatorname{j}\pi(M_{\mathrm{R}}-1)\sin(\theta^{(\operatorname{R})}_{\ell})}]^{\operatorname{T}},\\
    \label{eq:a_T}
    \bm{a}_{\operatorname{T}}(\theta^{(\operatorname{T})}_{\ell}) = [1,\mathrm{e}^{-\operatorname{j}\pi\sin(\theta^{(\operatorname{T})}_{\ell})},\ldotsTwo,\mathrm{e}^{-\operatorname{j}\pi(M_{\mathrm{T}}-1)\sin(\theta^{(\operatorname{T})}_{\ell})}]^{\operatorname{T}}
\end{align}
with subcarrier spacing $\Delta f$, symbol duration $\Delta T$, half wavelength antenna spacing, and number of subcarriers $M_{\mathrm{SC}}$, symbols $M_{\mathrm{SN}}$, and receive and transmit antennas $M_\mathrm{R}$ and $M_\mathrm{T}$, respectively. We do not explicitly specify statistical characteristics for any other channel parameters, as our subsequent findings apply in general.
In the remainder of this work, we use $\bm{H}$ and its vectorized version $\text{vec}(\bm{H})$ interchangeably.
\subsection{Statistical Independence in Bayesian Networks}
\label{sec:d_seperation}
A probabilistic graph corresponds to a graphical representation of a statistical model, in which random variables/vectors are modeled as nodes and statistical dependencies as edges. The causality between two dependent nodes, if clear, is encoded by a directed edge. If all edges in a probabilistic graph are directed, the graph is called a \ac{BN}. An example of a \ac{BN} is given in Fig. \ref{fig:setups} a). One advantage of explicitly representing a statistical model as a \ac{BN} is the possibility to directly infer conditional dependencies between nodes across the whole graph. To do so, two different setups have to be distinguished. The triplet $(A,C,B)$ in Fig. \ref{fig:setups} a) builds a v-structure due to both directed edges pointing towards $C$ (i.e., $A \rightarrow C \leftarrow B$). We assume that neither $A$ nor $B$ deterministically determines $C$. Then, the endpoints $A$ and $B$ in a v-structure are so-called d-separated if and only if neither the center node $C$ nor one of its descendants (i.e., $D$ in Fig. \ref{fig:setups} a)) is observed \cite[Sec. 3.3.1]{Koller2009}. If the d-separation is so-called sound, it implies statistical independence, which is typically the case and is assumed throughout this work \cite[Sec. 3.3.2]{Koller2009}. In any other configuration of arrows (e.g., $A \rightarrow C \rightarrow D$), the endpoints $A$ and $D$ are d-separated if and only if the center node $C$ is observed. If the two endpoints in every triplet of adjacent nodes in any trail in the \ac{BN} between two specific nodes of interest exhibit d-separation, these two nodes are d-separated and, thus, statistically independent.  

\begin{figure}[t]
    \includegraphics{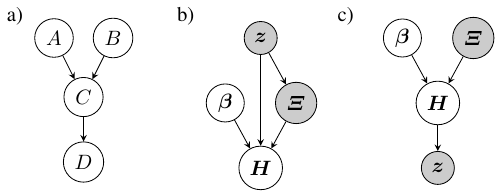}
	\caption{a) Exemplary Bayesian network, b) the sensing and modeling setup and c) the direct inference setup.}
	\label{fig:setups}
	\vspace{-0.6cm}
\end{figure}

\section{Main Result}
\label{sec:main_result}
Our characterization of side information $\bm{z}$ is based on the following observation. Preserving the zero channel mean and Toeplitz structured channel \acp{CM} by conditioning on $\bm{z}$ is solely linked to the impact of $\bm{z}$ on $\bm{\beta} = \{\beta_\ell\}_{\ell=1}^L$. Formally, this observation is presented in Theorem \ref{th:beta_main}.
\begin{theorem}
\label{th:beta_main}
Let $\bm{H}$ be defined according to \eqref{eq:channel_tensor} with Assumptions \ref{as:dl} and \ref{as:independence}, and steering vectors \eqref{eq:a_tau}-\eqref{eq:a_T}. Let $\bm{z}$ be any side information about $\bm{H}$. Moreover, let $\bm{\Xi}$ contain the channel parameters $\{p_\ell,\tau_\ell,\nu_\ell,\theta_\ell^{(\operatorname{R})},\theta_\ell^{(\operatorname{T})}\}_{\ell=1}^{L}$. Then, if 
\begin{equation}
    \label{eq:beta_l_uniform}
    \beta_\ell | (\bm{\Xi}, \bm{z}) \sim \mathcal{U}([0,2\pi])\ \text{for}\ \text{all}\ \ell = 1,\ldots,L
\end{equation}
holds true, it implies 
\begin{equation}
    \label{eq:preserve_hz}
    \E[\bm{H}|\bm{z}] = \bm{0},\E[\bm{H}\bm{H}^{\operatorname{H}}|\bm{z}] \in \bar{\mathcal{C}}_{\mathcal{T}}
\end{equation}
for any arbitrary distribution $p(\bm{\Xi})$ and with $\bar{\mathcal{C}}_{\mathcal{T}} = \mathcal{C}^{(M_{\mathrm{SC}})}_{\mathcal{T}} \otimes \mathcal{C}^{(M_{\mathrm{SN}})}_{\mathcal{T}} \otimes \mathcal{C}^{(M_{\mathrm{R}})}_{\mathcal{T}} \otimes \mathcal{C}^{(M_{\mathrm{T}})}_{\mathcal{T}}$, where $\mathcal{C}^{(F)}_{\mathcal{T}}$ denotes the set of $F \times F$ Toeplitz structured \acp{CM}. %The set $[-\pi,\pi]^L$ describes the L-ary cartesian power of $[-\pi,\pi]$.
\end{theorem}
\begin{proof}
See Appendix \ref{proof_beta_main}.
\end{proof}
Theorem \ref{th:beta_main} states that if the conditioning on $\bm{z}$ does not affect the statistical characteristics of $\beta_\ell$, then the channel's \ac{WSSUS} and zero-mean properties are preserved, independent of the statistical characteristics of the other channel parameters $\bm{\Xi}$ and any potential effects of $\bm{z}$ on these parameters. The cases where \eqref{eq:beta_l_uniform} is either true or false can be divided in just two distinct setups by means of a \ac{BN} representation, which leads to a comprehensive characterization of side information on its effect on $\beta_\ell$ and, thus, on the channel's \ac{WSSUS} and zero-mean properties. 

\subsection{Sensing and Modeling}
\label{sec:modelling_setup}
One setup is illustrated by the \ac{BN} given in Fig. \ref{fig:setups} b). It describes the situation in which the side information $\bm{z}$ is used to infer the channel parameters $\bm{\Xi}$ and/or the channel $\bm{H}$ without being directly observed through $\bm{H}$ itself. Since we condition on $(\bm{\Xi},\bm{z})$ in \eqref{eq:beta_l_uniform}, these variables are considered to be observed and marked gray. A multitude of different situations fall into this case, ranging from the co-existed design of joint communication and sensing with separate resources for communication and sensing functions \cite{Feng2021} to classical and modern \ac{ML}-based channel models based on, e.g., \acp{VAE} (cf. Section \ref{sec:channel_modelling}). 
The only two trails in Fig. \ref{fig:setups} b) between the observed $\bm{z}$ and $\bm{\beta} = \{\beta_\ell\}_{\ell=1}^L$ are $\bm{z} \rightarrow \bm{H} \leftarrow \bm{\beta}$ and $\bm{z} \rightarrow \bm{\Xi} \rightarrow \bm{H} \leftarrow \bm{\beta}$. Moreover, the only trails between the observed $\bm{\Xi}$ and $\bm{\beta}$ are given by $\bm{\Xi} \rightarrow \bm{H} \leftarrow \bm{\beta}$ and $\bm{\Xi} \leftarrow \bm{z} \rightarrow \bm{H} \leftarrow \bm{\beta}$. By applying the rules from Section \ref{sec:d_seperation} and assuming that $\bm{z}$ does not deterministically determine $\bm{H}$, we see that all trails contain a v-structure with non-observed center node $\bm{H}$ and no observed descendants. We conclude that $\bm{z}$ and $\bm{\beta}$ as well as $\bm{\Xi}$ and $\bm{\beta}$ are statistically independent. Hence, $p(\beta_\ell|\bm{\Xi},\bm{z})$ equals $p(\beta_\ell)$ for all $\ell=1,\ldots,L$ and due to \eqref{eq:beta_uniform}, \eqref{eq:beta_l_uniform} holds true in general.

\subsection{Direct Inference}
\label{sec:inference_setup}
The other setup describes the situation in which the side information $\bm{z}$ is a direct observation of the channel $\bm{H}$ itself. The arguably most important example for this setup is channel estimation, where $\bm{z}$ represents a noisy observation of the channel $\bm{H}$. The corresponding \ac{BN} is given in Fig. \ref{fig:setups} c). Since the center node $\bm{H}$ in $\bm{\beta} \rightarrow \bm{H} \rightarrow \bm{z} $ is not observed, $\bm{\beta}$ and $\bm{z}$ are not independent. Moreover, since the descendant $\bm{z}$ of the center node $\bm{H}$ in $\bm{\beta} \rightarrow \bm{H} \leftarrow \bm{\Xi}$ is observed, $\bm{\beta}$ and $\bm{\Xi}$ are also not independent. Thus, in this setup, both observed variables $(\bm{\Xi},\bm{z})$ potentially influence the statistics of $\bm{\beta}$ and we cannot claim \eqref{eq:beta_l_uniform} to hold true in general.

\section{Applications}

Section \ref{sec:main_result} provides the means to categorize any side information based on the way it is acquired and related to the channel parameters. In this section, we discuss possibilities how this characterization can be utilized to enhance or to analyze applications for wireless communication. 

\begin{figure}
    \centering
    \includegraphics{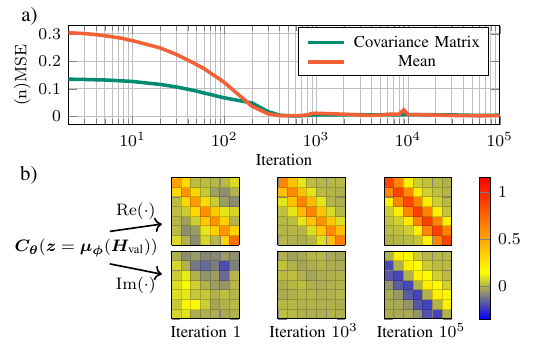}
\vspace{-0.5cm}
\caption{a) $\mathrm{nMSE}$ of the output covariance matrix to its Toeplitz projection and $\mathrm{MSE}$ of the output mean to zero over training iterations, b) real and imaginary parts of an exemplary output covariance matrix $\bm{C}_{\bm{\theta}}(\bm{z} = \bm{\mu}_{\bm{\phi}}(\bm{H}_{\text{val}}))$ generated by the same validation sample $\bm{H}_{\text{val}}$ after $1$, $10^3$ and $10^5$ iterations.}
\vspace{-0.6cm}
\label{fig:channel_modelling}
\end{figure}

\subsection{Channel Modeling}
\label{sec:channel_modelling}
One possible application is given in the context of modern \ac{ML}-based channel modeling, where generative models, e.g., \acp{GAN} \cite{Yang2019} or \acp{VAE}~\cite{baur2023}, are used to capture the underlying channel distribution from a training set of channel realizations. In this context, \acp{VAE} aim to encode the channel specific statistical features in a latent variable $\bm{z}$ by learning two distributions $q_{\bm{\phi}}(\bm{z}|\bm{H})$ and $p_{\bm{\theta}}(\bm{H}|\bm{z})$, where $\bm{H}$ is the channel (cf. \eqref{eq:channel_tensor}) and $(\bm{\phi},\bm{\theta})$ are \ac{NN} parameters. Typically, $p_{\bm{\theta}}(\bm{H}|\bm{z})$ and $q_{\bm{\phi}}(\bm{z}|\bm{H})$ are modeled as conditionally Gaussian, i.e., $p_{\bm{\theta}}(\bm{H}|\bm{z}) = \mathcal{N}_{\mathbb{C}}(\bm{H};\bm{\mu}_{\bm{\theta}}(\bm{z}),\bm{C}_{\bm{\theta}}(\bm{z}))$ and $q_{\bm{\phi}}(\bm{z}|\bm{H}) = \mathcal{N}(\bm{z};\bm{\mu}_{\bm{\phi}}(\bm{H}),\text{diag}(\bm{\sigma}_{\bm{\phi}}(\bm{H})^2))$. The \ac{VAE}'s training is based on forwarding training channels $\bm{H}_{\text{train}}$ through an encoder-decoder processing chain yielding $(\bm{\mu}_{\bm{\phi}}(\bm{H}_{\text{train}}),\bm{\sigma}_{\bm{\phi}}(\bm{H}_{\text{train}}))$ as well as $(\bm{\mu}_{\bm{\theta}}(\tilde{\bm{z}}),\bm{C}_{\bm{\theta}}(\tilde{\bm{z}}))$ with $\tilde{\bm{z}} \sim q_{\bm{\phi}}(\bm{z}|\bm{H}_{\text{train}})$\footnote{We refer to \cite{baur2023} for more details about \acp{VAE} for wireless communication.}. Due to the strong fluctuations of $\beta_\ell$ (cf. Section \ref{sec:channel_model}), these phases cannot contain distinct statistical channel characteristics rendering it unlikely that the \ac{VAE} stores $\beta_\ell$-specific information in $\bm{z}$. In consequence, the \ac{BN} in Fig. \ref{fig:setups} b) applies, \eqref{eq:beta_l_uniform} holds, and $\bm{\mu}_{\bm{\theta}}(\bm{z})$ and $\bm{C}_{\bm{\theta}}(\bm{z})$ are expected to be learned to be zero and Toeplitz in any domain, respectively (cf. Section \ref{sec:modelling_setup}). Thus, we can utilize Theorem \ref{th:beta_main} as a tool to verify the \ac{VAE}'s correct training. In Fig. \ref{fig:channel_modelling}, the behaviour of the parameterized conditional mean and \ac{CM} during the \ac{VAE}'s training is shown. As training set we used $50\,000$ narrowband and static channels generated by the geometry-based stochastic channel model QuaDRiGa \cite{Jaeckel2014}, which are randomly sampled in a $120\,^{\circ}$ sector of the ``3GPP\_38.901\_UMa\_NLOS'' scenario. The \ac{BS} and each user is equipped with $8$ and $1$ antennas, respectively, which results in $\bm{H}$ exhibiting solely the spatial receiver domain. We allowed $\bm{C}_{\bm{\theta}}(\bm{z})$ to take any unstructured \ac{CM} during training by outputting its arbitrarily learnable Cholesky decomposition $\bm{L}_{\bm{\theta}}(\bm{z})$ with $\bm{C}_{\bm{\theta}}(\bm{z}) = \bm{L}_{\bm{\theta}}(\bm{z})\bm{L}_{\bm{\theta}}(\bm{z})^{\operatorname{H}}$. Similary, we also allowed $\bm{\mu}_{\bm{\theta}}(\bm{z})$ to take any value. In Fig. \ref{fig:channel_modelling} a), the \ac{NMSE} $\mathrm{nMSE} = 1/500 \sum_{n=1}^{500}(\|\bm{C}_{\bm{\theta}}(\bm{z}_n) - \bm{C}^{(\operatorname{P})}_{\bm{\theta}}(\bm{z}_n)\|_{\operatorname{F}}^2)/\|\bm{C}_{\bm{\theta}}(\bm{z}_n)\|_{\operatorname{F}}^2$ of output covariances matrices $\bm{C}_{\bm{\theta}}(\bm{z}_n)$ to their orthogonal projection $\bm{C}^{(\operatorname{P})}_{\bm{\theta}}(\bm{z}_n)$ on Toeplitz matrices (cf. \cite{Grigoriadis1994}) over the training iterations is shown. Additionally, the \ac{MSE} $\mathrm{MSE} = 1/500 \sum_{n=1}^{500}\|\bm{\mu}_{\bm{\theta}}(\bm{z}_n) - \bm{0}\|_2^2$ of output means $\bm{\mu}_{\bm{\theta}}(\bm{z}_n)$ to zero is given. Both are generated from randomly sampled latent realizations $\{\bm{z}_n\}_{n=1}^{500}$ after each iteration. It can be seen that both measures decrease significantly, approaching zero. Remarkably, Fig. \ref{fig:channel_modelling} a) shows that the \ac{VAE}, although not explicitly constrained, is trained to solely output Toeplitz structured \acp{CM} and zero means. This behaviour indicates correct training in the sense that $\bm{z}$ captures distribution relevant features and encodes no information about $\bm{\beta}$. Additionally, in Fig. \ref{fig:channel_modelling} b), the convergence of $\bm{C}_{\bm{\theta}}(\bm{z})$ towards Toeplitz structured matrices is illustrated by the real and imaginary part of one exemplary output \ac{CM}, where the conditional mean $\bm{\mu}_{\bm{\phi}}(\bm{H}_{\text{val}})$ of $q_{\bm{\phi}}(\bm{z}|\bm{H}_{\text{val}})$ for a fixed validation channel $\bm{H}_{\text{val}}$ is used to generate the plotted $\bm{C}_{\bm{\theta}}(\bm{z} = \bm{\mu}_{\bm{\phi}}(\bm{H}_{\text{val}}))$. These findings also theoretically underpin results from \cite{baur2023} and \cite{Fesl2024}, where it is observed that constraining $\bm{C}_{\bm{\theta}}(\bm{z})$ to be Toeplitz (or circulant as Toeplitz approximation) or $\bm{\mu}_{\bm{\theta}}(\bm{z})$ to be zero results in \acp{VAE} with strong channel estimation performance.

\subsection{Channel Clustering}
\begin{figure}[t]
    \centering
    \includegraphics{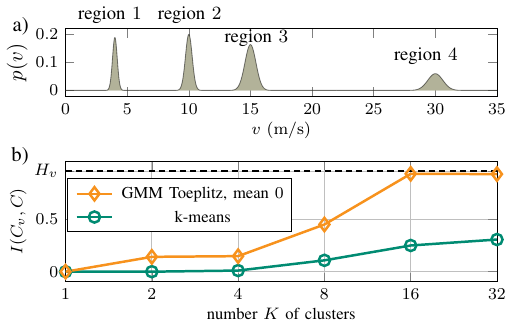}
\caption{a) Velocity probability distribution with four distinct regions, b) mutual information $I(C_v,C)$ between the GMM Toeplitz and zero mean ($C_g$) or k-means ($C_k$) clustering and the ground truth velocity clustering $C_v$ with $C \in \{C_g,C_k\}$.}
\label{fig:clustering}
	\vspace{-0.6cm}
\end{figure}

In the previous section, we utilized Theorem \ref{th:beta_main} to verify correct training and to interpret the information encoded in the \ac{VAE}'s latent embedding. In this section, we actively build on Theorem \ref{th:beta_main} to directly regularize a clustering algorithm towards the desired outcome. 
We consider a situation, in which we aim to cluster time-varying channel trajectories according to the users' velocities without having access to any explicit velocity information. In this example, $\bm{z}$ represents the discrete and finite clusters. If $\bm{z}$ solely encodes the velocities and, thus, information about $\bm{\Xi}$, the \ac{BN} in Fig. \ref{fig:setups} b) applies, and \eqref{eq:beta_l_uniform} holds (cf. Section \ref{sec:modelling_setup}). In consequence, we know $\E[\bm{H}\bm{H}^{\operatorname{H}}|\bm{z} = \text{cluster}\ i]$ to be Toeplitz in any domain and the mean $\E[\bm{H}|\bm{z} = \text{cluster}\ i]$ to be zero for all $i$ in advance (cf. Theorem \ref{th:beta_main}), and, thus, any mean-based clustering algorithm (e.g., k-means) to be sub-optimal. On the other hand, while \acp{GMM} can also be used as generative models as in Section \ref{sec:channel_modelling}, \ac{GMM}-based clustering allows to incorporate the insights of Theorem \ref{th:beta_main}. Specifically, \acp{GMM} assign an unconstrained Gaussian distribution to every cluster individually. However, it is also possible to regularize their means and \acp{CM}. In the following, we utilize the \ac{GMM} covariance Toeplitz parameterization and clustering from \cite{Turan2023}, in which the cluster indices are used for \ac{CSI} feedback\footnote{We refer to \cite{Turan2023Journal} for more details about \acp{GMM} for wireless communication.}. However, in addition to \cite{Turan2023}, we also enforce the \ac{GMM} means to be zero, such that every cluster is regularized to have a zero mean and Toeplitz structured \ac{CM}.  We trained this regularized \ac{GMM} with $80\,000$ \ac{SISO} narrowband time-varying channels, which are sampled $16$ times every $0.5$ ms in a $120\,^{\circ}$ sector of the ``3GPP\_38.901\_UMa\_NLOS'' scenario of QuaDRiGa. This results in $\bm{H}$ (cf. \eqref{eq:channel_tensor}) exhibiting solely the temporal domain. Each user's velocity is randomly drawn from the density $p(v)$ illustrated in Fig. \ref{fig:clustering} a). This density exhibits four distinct velocities regions (region 1-4), which allows a ground truth velocity clustering of a test set of $8000$ channel trajectories, where each trajectory is labeled with the corresponding user's velocity region. 
After training the regularized \ac{GMM}, which gets no direct velocity information during training as well as testing, we cluster the test set using the \ac{GMM}. For evaluation, we then compare the ground truth velocity clustering represented by the random variable $C_v$ with the \ac{GMM} clustering represented by $C_g$ by means of their mutual information $I(C_v,C_g)$ (cf. \cite{Vinh2010}). The same is done for k-means clustering $C_k$ and is illustrated for different k-means and \ac{GMM} cluster numbers $K$ in Fig. \ref{fig:clustering} b). It can be seen that the \ac{GMM} clustering generally shows higher mutual information with the ground truth velocity clustering than k-means. Remarkably, although the \ac{GMM} does not get any explicit velocity information, its mutual information achieves the entropy $H_v$ of the ground truth velocity clustering for $K=16$ and larger. This implies that in this regime, each \ac{GMM} cluster is associated with exactly one of the velocity regions 1-4 in Fig. \ref{fig:clustering} a), i.e., the regularized \ac{GMM} yields perfect velocity clustering. In the same $K$-regime, the mutual information of k-means shows a significant gap to $H_v$.

\begin{figure}[t]
    \centering
    \includegraphics{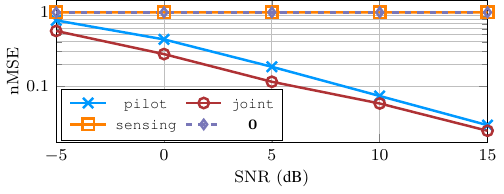}
\caption{$\mathrm{nMSE}$ over the SNR of the four channel estimators \texttt{pilot}, \texttt{sensing}, \texttt{joint} and the zero vector.}
\label{fig:estimation}
	\vspace{-0.6cm}
\end{figure}
\begin{figure*}[t]
\begin{align}
\label{eq:conditional_mean}
\E_{\bm{H}|\bm{z}}[\bm{H}] = & \E_{(\bm{\beta},\bm{\Xi})|\bm{z}}[\bm{H}] =
\E_{\bm{\Xi}|\bm{z}}\left[\sum_{\ell=1}^L \sqrt{p_\ell} \E_{\bm{\beta}|(\bm{z},\bm{\Xi})}\left[\mathrm{e}^{-\operatorname{j}\beta_\ell}\right]\bm{a}_{f,\ell} \otimes \bm{a}_{t,\ell} \otimes \bm{a}_{\operatorname{R},\ell} \otimes \bm{a}_{\operatorname{T},\ell}\right] \\
\label{eq:conditional_cov}
\E_{\bm{H}|\bm{z}}[\bm{H}\bm{H}^{\operatorname{H}}] = & \E_{\bm{\Xi}|\bm{z}}\left[\sum_{\ell,\ell'=1}^L \sqrt{p_\ell}\sqrt{p_{\ell'}} \E_{\bm{\beta}|(\bm{z},\bm{\Xi})}\left[\mathrm{e}^{-\operatorname{j}\beta_\ell}\mathrm{e}^{\operatorname{j}\beta_{\ell'}}\right]\bm{a}_{f,\ell}\bm{a}_{f,\ell'}^{\operatorname{H}} \otimes \bm{a}_{t,\ell}\bm{a}_{t,\ell'}^{\operatorname{H}} \otimes \bm{a}_{\operatorname{R},\ell}\bm{a}_{\operatorname{R},\ell'}^{\operatorname{H}} \otimes \bm{a}_{\operatorname{T},\ell}\bm{a}_{\operatorname{T},\ell'}^{\operatorname{H}}\right] 
\end{align}
\hrule
\vspace{-0.4cm}
\end{figure*}
\subsection{Channel Estimation}

A further application of Section \ref{sec:main_result} is given in the context of channel estimation. Theorem \ref{th:beta_main} characterizes the conditional mean $\E[\bm{H}|\bm{z}]$ of the channel $\bm{H}$ given some side information $\bm{z}$. It is well known that this conditional mean represents the \ac{MMSE} channel estimator based on $\bm{z}$. Thus, Theorem \ref{th:beta_main} in combination with the setups in Fig. \ref{fig:setups} establishes a framework to analyze which kind of information can be fundamentally utilized for estimating the channel. More precisely, if $\bm{z}$ is not statistically relevant for $\beta_\ell$, then the best channel estimate given $\bm{z}$ is the zero vector (cf. Theorem \ref{th:beta_main}). Consequently, $\bm{z}$ has to contain a descendant of $\bm{H}$, i.e., a pilot observation for channel estimation (cf. Fig. \ref{fig:setups} c)). Additionally, Section~\ref{sec:main_result} also theoretically underpins that as long as $\bm{z}$ contains some direct observation of $\bm{H}$, any additional side information about $\bm{\Xi}$ can improve the estimation. This is due to $\beta_\ell$ being able to statistically depend on $\bm{\Xi}$ if and only if $\bm{H}$ or one of its descendants (i.e, a direct observation) is given (cf. Section \ref{sec:d_seperation}). For evaluating these insights, we trained several fully-connected \acp{NN} for channel estimation in an end-to-end fashion. Specifically, we generated $10^5$ channels using \eqref{eq:channel_tensor} containing solely the spatial receiver domain of dimension $16$. All these realizations contain three paths, $\beta_\ell$ and $\theta_\ell^{(\operatorname{R})}$ are  uniformly distributed, and the path losses $\sqrt{p_\ell}$ are uniformly distributed between zero and one and then normalized to sum up to one. In Fig. \ref{fig:estimation}, the estimation performance in terms of their $\mathrm{nMSE}$ to the ground truth channel of three different \ac{NN}-based estimators over the \ac{SNR} is shown, denoted by \texttt{sensing}, \texttt{pilot} and \texttt{joint}. For evaluation, we used a test set of $5000$ channels. All \acp{NN} are trained by minimizing the \ac{MSE} between their output and the true channel, and their \ac{NN} depth and width were optimized by a random search using a validation set of size $5000$. As input for the \texttt{sensing} \ac{NN}, we took the ground truth steering vectors (cf. \eqref{eq:a_R}) of the corresponding channel's three paths. However, since these have no statistical relevance for $\beta_\ell$, \eqref{eq:beta_l_uniform} holds, the \ac{MMSE} estimator is given by the zero vector and the \texttt{sensing} network does not outperform the zero vector. On the other hand, the \texttt{pilot} network takes a noisy channel observation as input. Thus, a descendant of $\bm{H}$ is observed, Fig.~\ref{fig:setups}~c) applies and the estimator outperforms the zero vector (cf. Section \ref{sec:inference_setup}). The \texttt{joint} network takes both, the ground truth steering vectors as well as the noisy channel observation as input. Since a descendant of $\bm{H}$ (i.e., the noisy channel observation) is given, $\bm{\beta}$ and $\bm{\Xi}$ can statistically depend on each other (cf. Section \ref{sec:inference_setup}), and additional information about $\bm{\Xi}$ can improve the estimation, which is seen in Fig. \ref{fig:estimation} in form of \texttt{joint} outperforming \texttt{pilot}.

\section{Conclusion}

In this work, we introduced a comprehensive framework, which combines insights from a generic channel model with \acp{BN} to categorize side information on its effect on the channel's \ac{WSSUS} and zero-mean properties. This framework can be utilized in various ways. We demonstrated how it can be leveraged to analyze side information for channel generation or estimation, and to directly regularize channel clustering. While we discussed three particular applications, this analysis indicates that many more exist, which are part of future work.

\appendix

\subsection{Proof of Theorem \ref{th:beta_main}}
\label{proof_beta_main}
To prove Theorem \ref{th:beta_main}, we reformulate the conditional mean of $\bm{H}$ given $\bm{z}$ according to \eqref{eq:conditional_mean}, where we insert the definition \eqref{eq:channel_tensor} of $\bm{H}$. For simplicity, we also leave out the arguments of the steering vectors. The inner expectation $\E_{\bm{\beta}|(\bm{z},\bm{\Xi})}[\mathrm{e}^{-\operatorname{j}\beta_\ell}]$ equals zero according to the assumption $\beta_\ell | (\bm{\Xi}, \bm{z}) \sim \mathcal{U}([0,2\pi])$ in Theorem \ref{th:beta_main}. Equivalently, the conditional \ac{CM} $\E[\bm{H}\bm{H}^{\operatorname{H}}|\bm{z}]$ can be decomposed according to \eqref{eq:conditional_cov}. Considering assumption \eqref{eq:para_assumption} and the above reasoning, the inner expectation $\E_{\bm{\beta}|(\bm{z},\bm{\Xi})}[\mathrm{e}^{-\operatorname{j}\beta_\ell}\mathrm{e}^{\operatorname{j}\beta_{\ell'}}]$ equals zero for $\ell \neq \ell'$ and one otherwise. Since $\bm{a}_{(\cdot)}(\cdot)\bm{a}_{(\cdot)}(\cdot)^{\operatorname{H}}$ is a Toeplitz structured matrix for any domain and any channel parameter configuration (cf. \eqref{eq:a_tau}-\eqref{eq:a_T}), it is Toeplitz in expectation concluding the proof.

\bibliographystyle{IEEEtran.bst}
\bibliography{references.bib}

\end{document}